\documentclass[a4paper,UKenglish,cleveref, autoref, thm-restate]{lipics-v2021}

\nolinenumbers 

\usepackage[utf8]{inputenc}
\usepackage{amsmath,amssymb,amsfonts}
\usepackage{physics}
\usepackage{enumerate}
\usepackage{mathtools}

\usepackage{cite}
\usepackage{subcaption}

\usepackage[disable]{todonotes} 
\usepackage{letltxmacro}
\LetLtxMacro{\todom}{\todo}

\newcommand{\monomer}[1]{\textbf{\lowercase{#1}}}
\newcommand{\polymer}[1]{\textbf{\uppercase{#1}}}

\newcommand{\functionproblem}[1]{\textsc{#1}}

\renewcommand{\vec}[1]{\mathbf{#1}}

\newcommand{\defn}[1]{\emph{#1}}

\newcommand{\N}{\ensuremath{\mathbb{N}}}

\newcommand{\saturatedconfigs}[1]{\Gamma_{#1}}

\newcommand{\selfsaturated}{self-saturated}

\newcommand{\Count}{\ensuremath{\mathsf{Count}}}
\newcommand{\Exists}{\ensuremath{\mathsf{Exists}}}
\newcommand{\Tied}{\ensuremath{\mathsf{Tied}}}

\usepackage{graphicx}
\usepackage[normalem]{ulem} 

\definecolor{specialblue}{rgb}{0.0, 0.53, 0.71}

\title{Computing properties of thermodynamic binding networks: An integer programming approach}

\funding{Supported by NSF award 1900931 and CAREER award 1844976.}

\titlerunning{Computing properties of thermodynamic binding networks: An IP approach}

\author{David Haley}{University of California, Davis}{drhaley@ucdavis.edu}{}{}
\author{David Doty}{University of California, Davis}{doty@ucdavis.edu}{}{}

\authorrunning{D. Haley and D. Doty}
 
\Copyright{David Haley and David Doty} 



\category{} 

\relatedversion{} 

\supplement{
\emph{GitHub repo for software:}
    \url{https://github.com/drhaley/stable_tbn}
}

\begin{document}

\maketitle

\begin{abstract}
    The thermodynamic binding networks (TBN) model~\cite{tbn} is a tool for studying
    engineered molecular systems.
    The TBN model
    allows one to reason about their behavior through a simplified abstraction that ignores details about molecular composition, 
    focusing on two key determinants of a system's energetics common to \emph{any} chemical substrate: 
    how many molecular bonds are formed, 
    and how many separate complexes
    exist in the system.
    We formulate as an integer program the $\mathsf{NP}$-hard problem of computing \emph{stable} 
    (a.k.a., minimum energy) 
    configurations of a TBN:
    those configurations that maximize the number of bonds and complexes.
    We provide open-source software \cite{stable_tbn_software} solving this integer program.
    We give empirical evidence that this approach enables dramatically faster computation of TBN stable configurations than previous approaches based on SAT solvers~\cite{tbn-sat}.  
    Furthermore, unlike SAT-based approaches, our integer programming formulation can reason about TBNs in which some molecules have unbounded counts.
    These improvements in turn allow us to efficiently automate verification of desired properties of practical TBNs.
    Finally, we show that 
    the TBN has a natural representation with a unique Hilbert basis
    describing
    the ``fundamental components'' out of which \emph{locally} minimal energy configurations are composed. 
    This characterization helps verify correctness of not only stable configurations, but entire ``kinetic pathways'' in a TBN.

\end{abstract}


\section{Introduction}


Recent experimental breakthroughs in DNA nanotechnology~\cite{dnaNanoSurveySeelig2015} have enabled the construction of intricate molecular machinery whose complexity rivals that of biological macromolecules,
even executing general-purpose algorithms~\cite{drmaurdsa}.
A major challenge in creating synthetic DNA molecules that undergo desired chemical reactions is the occurrence of erroneous ``leak''  reactions~\cite{qian2011scaling},
driven by the fact that the products of the leak reactions are more energetically favorable.
A promising design principle to mitigate such errors is to build ``thermodynamic robustness'' into the system, ensuring that leak reactions incur an energetic cost~\cite{thachuk2015leakless, wang2018effective,minev2021robust} by logically forcing 
one of two unfavorable events:
either many molecular bonds must break---an ``enthalpic'' cost---or many separate molecular complexes 
(called \emph{polymers} in this paper) 
must simultaneously come together---an ``entropic'' cost.

The model of \emph{thermodynamic binding networks} (TBNs)~\cite{tbn} was defined as a combinatorial abstraction of such molecules, 
deliberately simplifying
substrate-dependent details of DNA in order to isolate the foundational energetic contributions of forming bonds and separating polymers.
A TBN consists of \emph{monomers}
containing
specific \emph{binding sites}, where binding site $a$ can bind only to its complement $a^*$.
A key aspect of the TBN model is the lack of geometry: a monomer is an \emph{unordered} collection of binding sites such as $\{a,a,b^*,c\}$.
A \emph{configuration} of a TBN describes which
monomers are grouped into \emph{polymers}; bonds can only form within a polymer.
One can formalize the ``correctness'' of a TBN by requiring that its desired configuration(s) be \emph{stable}:
the configuration maximizes the number of bonds formed,
a.k.a., it is \emph{saturated},
and, among all saturated configurations, 
it maximizes the number of separate polymers.\footnote{
    This definition captures the limiting case (often approximated in practice in DNA nanotechnology) corresponding to increasing the strength of bonds, while diluting (increasing volume), such that the ratio of binding to unbinding rate goes to infinity.
}
See~\cref{fig:tbn-intro-example} for an example.
Stable configurations are meant to capture the \emph{minimum free energy structures} of the TBN.
Unfortunately, answering basic questions such as 
``\emph{Is a particular TBN configuration stable?}''
turn out to be $\mathsf{NP}$-hard~\cite{tbn-sat}.

\begin{figure}[h]
    \vspace{-0.2cm}
    \centering
    \includegraphics[width=4.7in]{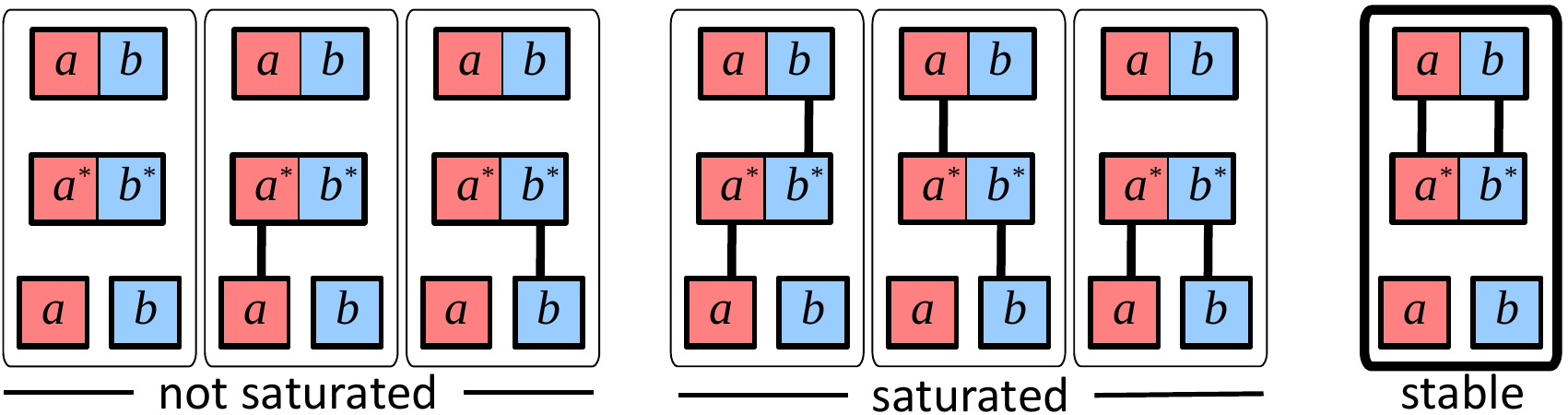}
    \vspace{-0.2cm}
    \caption{ \footnotesize
    Example of a simple thermodynamic binding network (TBN). 
    There are four monomers: 
    $\monomer{m}_1 = \{a^*, b^*\}, 
    \monomer{m}_2 = \{a, b\}, 
    \monomer{m}_3 = \{a\}, 
    \monomer{m}_4 = \{b\},$
    with seven configurations shown:
    four of these configurations are saturated because they have the maximum of 2 bonds.
    Of these, three have 2 polymers and one has 3 polymers, making the latter the only stable configuration.
    Despite the suggestive lines between binding sites, the model of this paper ignores individual bonds, defining a configuration solely by how it partitions the set of monomers into polymers,
    assuming that a maximum number of bonds will form within each polymer.
    (Thus other configurations exist besides those shown, which would merge polymers shown without allowing new bonds to form.)
    }
    \label{fig:tbn-intro-example}
    \vspace{-0.6cm}
\end{figure}

\subsection{Our contribution}

Our main contribution is a reduction that formulates the problem of finding stable configurations of a TBN as an integer program (IP).
Of course, the problem, appropriately formalized, is ``clearly'' an $\mathsf{NP}$ search problem, so the mere existence of such a reduction is not particularly novel.
However, our formulation is notable in 
three respects:
1) We carefully avoid certain symmetries (particularly those present in the existing SAT-based formulation of Breik et al.~\cite{tbn-sat}), which dramatically increases the search efficiency in practice.
2) We use the optimization ability of IP solvers as a natural way to maximize the number of polymers in any saturated configuration.
3) Our formulation leads to a natural interpretation of the \emph{Hilbert basis}~\cite{de2012algebraic} of a TBN as its minimal saturated polymers, which intuitively are the polymers existing in any \emph{local} energy minimum configuration.
Since highly optimized software exists for calculating Hilbert bases~\cite{4ti2}, 
this expands the range of TBN behaviors that can be automatically reasoned about.

This formulation allows us to automate portions of formal reasoning about TBNs, helping verify their correctness.
The TBN model abstracts away the continuous nature of real free energy into discrete integer-valued steps. 
In the limit of dilute solutions 
(bringing together polymers incurs a large energy cost)
and very strong bonds 
(breaking a bond incurs a huge energy cost),
even one integer step of energy difference is considered significant.
Thus it is crucial for verifying such systems that we identify the \emph{exact} solution to the optimization problem, rather than settling for more efficiently computable approximations (e.g., via continuous relaxation \cite{conforti2014integer} or local search \cite{shaw2003constraint}).


\subsection{Related work}
Breik, Thachuk, Heule, and Soloveichik~\cite{tbn-sat} characterize the computational complexity of several natural problems related to TBNs.
For instance, it is $\mathsf{NP}$-complete to decide whether a saturated configuration exists with a specified number of polymers, and even $\mathsf{NP}$-hard to approximate within any constant factor the number of polymers in a stable configuration (i.e., the maximum in any saturated configuration).

Breik et al. also developed software using a SAT solver to produce stable configurations of a TBN. 
This formulation requires ``labelled'' monomers (where two different instances of the same monomer type are represented by separate Boolean variables), 
which become nodes in a graph, and polymers are realized as connected components within the graph.
By labelling the monomers they become unique copies of the same monomer type;
$n$ copies of a monomer type increases the size of the search space by factor $n!$ by considering these symmetric configurations separately.
Furthermore, the software explicitly explores all possible symmetries of bonding arrangements within a polymer.
For instance, monomers $\{a^*, a^*\}$ and $\{a,a\}$ can bind in two different ways 
(the first $a^*$ can bind either the first or second $a$),
even though both have the same number of bonds and polymers.
This over-counting of symmetric configurations prevents the software from scaling to efficiently analyze certain TBNs with large counts of monomers.
Our IP formulation avoids both types of symmetry.


\section{Preliminaries}

\subsection{Definitions}


A \defn{multiset} is an unordered collection of objects allowing duplicates 
(including countably infinite multiplicities), 
e.g., $\vec{v} = \{2\cdot a, b, \infty \cdot d\}$.
Equivalently, a multiset with elements from a finite set $U$ is a vector $\vec{v} \in (\N \cup \{\infty\})^{U}$ describing the counts, indexed by $U$;
in the example of~\cref{fig:tbn-intro-example}, if $U = \{a,b,c,d\}$,
then
$\vec{v}(a) = 2$, 
$\vec{v}(b) = 1$,
$\vec{v}(c) = 0$,
and
$\vec{v}(d) = \infty$.
The \defn{cardinality} of a multiset $\vec{v} \in \N^U$ is $|\vec{v}| = \sum_{u \in U} \vec{v}(u)$;
a \defn{finite multiset} $\vec{v}$ obeys $|\vec{v}| < \infty$.
%
A \defn{site type} is a formal symbol, such as $a$,
representing a specific binding site on a molecule;
in~\cref{fig:tbn-intro-example} the site types are
$a, a^*, b, b^*$.
%
Each site type has a corresponding \defn{complement type} denoted by a star, e.g. $a^*$.
Complementarity is an involution:  $(a^*)^* = a$.  
A site and its complement can form an attachment called a \defn{bond}.
We follow the convention that for any complementary pair of sites $a,a^*$,
the total count of $a^*$ across the whole TBN is at most that of $a$, i.e., the starred sites are \defn{limiting}.
A \defn{monomer type} is a finite multiset of site types.  
When context implies a single instance of a monomer/site type, we may interchangeably use the term \defn{monomer/site}.\footnote{
    Concretely, in a DNA nanotech design, a monomer corresponds to a strand of DNA, whose sequence is logically partitioned into binding sites corresponding to long (5-20 base) regions, e.g., $5'$-AAAGG-$3'$, intended to bind to the complementary sequence $3'$-TTTCC-$5'$ that is part of another strand.
}

A \defn{thermodynamic binding network (TBN)} is a multiset of monomer types;
equivalently, a vector in $(\N \cup \{\infty\})^m$,
if $m$ is the number of monomer types and we have fixed some standardized ordering of them.  
We allow some monomer counts to be infinite in order to capture the case where some monomers are added in ``large excess'' over others, a common experimental approach~\cite{rothemund2006folding, qian2011scaling}.
A \defn{polymer} is a finite multiset of monomer types,
equivalently, a vector in $\N^m$,
where $m$ is the number of monomer types in a standardized ordering.\footnote{
    The term ``polymer'' is chosen to convey the concept of combining many atomic objects into a complex, but it is not necessarily a linear chain of repeated units.
}
Note that despite the suggestive lines representing bonds in~\cref{fig:tbn-intro-example}, 
this definition does not track which pairs of complementary sites are bound within a polymer.

Let $S_\mathcal{T}$ (respectively, $S^*_\mathcal{T}$) be the set of unstarred (resp., starred) site types of $\mathcal{T}$.
For a monomer $\monomer{m}$ and site type $s \in S_\mathcal{T}$, let $\monomer{m}(s)$ denote the count of $s$ minus the count of $s^*$ in $\monomer{m}$ (intuitively, $\monomer{m}(s)$ is the ``net count'' of $s$ in $\monomer{m}$, negative if there are more $s^*$.)
For $s^* \in S_\mathcal{T}^*$ let $\monomer{m}(s^*) = -\monomer{m}(s)$.
The \defn{exposed sites} of a polymer $\polymer{P}$ are a finite multiset of site types 
that results from removing as many (site, complement) pairs from a polymer as possible, described by the net count of sites when summed across all monomers in the polymer. 
For example, in the polymer 
$\{\{a^*, b^*, c^*\}, \{a, c\}, \{a, b, c\}, \{c,d^*\}\}$,
the exposed sites are $\{a, 2\cdot c, d^*\}$.
%

A \defn{configuration} of a TBN is a partition of the monomers of the TBN into polymers.
A polymer is \defn{\selfsaturated} if it has no exposed starred sites.
A configuration is \defn{saturated} if 
all of its polymers are \selfsaturated.
Since we assume that, across the entire configuration, starred sites are limiting,
this is equivalent to stipulating that the maximum possible number of bonds are formed.
Write $\saturatedconfigs{\mathcal{T}}$ to denote the set of all saturated configurations of the TBN $\mathcal{T}$. 
%
%
A configuration is \defn{stable} if it is saturated and has the maximum number of non-singleton polymers among all saturated configurations. 

Since the number of polymers may be infinite, we will use the equivalent notion that stable configurations are those that can be constructed by starting with the ``melted'' configuration whose polymers are all singletons containing only one monomer, performing the minimum number of polymer merges necessary to reach a saturated configuration.
For example, consider the TBN consisting of monomer types $\monomer{t} = \{a\}$, $\monomer{b} = \{a^*\}$,
with counts
$\infty \cdot \monomer{t}$ and
$2 \cdot \monomer{b}$.
The unique stable configuration has polymers 
$\{2 \cdot \{\monomer{b},\monomer{t}\}, \infty \cdot \monomer{t} \}$,
since two merges of a $\monomer{b}$ and a $\monomer{t}$ are necessary and sufficient to create this configuration from the individual monomers.

\subsection{Solvers}

The problems addressed in this paper are 
$\mathsf{NP}$-hard. 
To tackle this difficulty, 
we cast the problems as integer programs 
and use the publicly available IP solver SCIP~\cite{GamrathEtal2020OO}.

We also use the open-source software OR-tools\cite{ortools},
which is a common front-end for SCIP\cite{GamrathEtal2020OO}, Gurobi~\cite{gurobi}, and a bundled constraint programming solver CP-SAT.
Though we model our problems as IPs, we would also like to be able to solve for all feasible/optimal solutions 
rather than just one,
which CP-SAT can do.
This flexible front-end lets us switch seamlessly between the two types of solvers without significant alterations to the model.

\todo{DD: There was a paragraph about SCIP for finding optimal value and CP-SAT for enumerating solutions, which seemed out of place, since we discuss this issue extensively in Paragraph~\ref{subsec:ip-find-optimal-then-cp-enumerate}, and this section is just supposed to point the reader to what tools we used.}
We use the software package 4ti2\cite{4ti2} to calculate Hilbert Bases as described in~\cref{sec:hilbert}.

\section{Computing stable configurations of TBNs}
\cref{subsec:stableconfigs-def} formally defines the stable configurations problem.
\cref{subsec:formulation} explains our IP formulation of the problem.
\cref{subsec:benchmarks} shows empirical runtime benchmarks.

\subsection{Finding stable configurations of TBNs}\label{subsec:stableconfigs-def}


We consider the problem
of finding the stable configurations of a TBN.  
Given a TBN $\mathcal{T}$, let $\Gamma_{\mathcal{T}}$ denote the set of all saturated configurations of $\mathcal{T}$.

Recall that a configuration $\gamma \in \Gamma_{\mathcal{T}}$ is defined as a partition of the monomers of $\mathcal{T}$ into polymers, so its elements $\polymer{P} \in \gamma$ are polymers, i.e., multisets of monomers.
For any $\gamma \in \Gamma_{\mathcal{T}}$, we define the corresponding 
\defn{partial configuration}
$\overline{\gamma} = \left\{\polymer{P} \in \gamma: |\polymer{P}| > 1\right\}$
that excludes polymers consisting of only a single monomer.
Note that in the context of $\mathcal{T}$, the mapping $\gamma \mapsto \overline{\gamma}$ is one-to-one.
We consider only partial configurations with finite-sized polymers. The notion of partial configuration will be useful in reasoning about TBNs with infinite monomer counts but finite size polymers, since all but finitely many monomers will be excluded from the partial configurations we consider.

Now we define the number of elementary merge operations required to reach a saturated configuration from the configuration of all singletons:
\begin{equation} \label{eq:count-merges-to-stable}
    m(\gamma) =
    \qty(\sum_{\polymer{P} \in \overline{\gamma}} |\polymer{P}|)
    -
    \abs{\overline{\gamma}}
\end{equation}
We can then define the stable configurations as those 
saturated configurations
that minimize the number of merges required to reach them from the all-singletons configuration.
\begin{equation*}
    \functionproblem{StableConfigs}(\mathcal{T}) = \left\{\gamma \in \Gamma_{\mathcal{T}}:  
    (\forall \gamma' \in \Gamma_{\mathcal{T}})\,\, m(\gamma) \leq m(\gamma')\right\}
\end{equation*}
Note that $m(\gamma) = m(\overline{\gamma})$.  
Thus
the \functionproblem{StableConfigs} problem may be equivalently posed as finding the set of partial configurations $\overline{\gamma}$ that minimize $m(\overline{\gamma})$.

We now describe how to handle infinite counts.
A configuration is saturated if and only if none of its starred sites (elements of $S_\mathcal{T}^*$) are exposed. 
Thus we focus on the subset of monomers that contain starred sites: 
the \emph{limiting monomers} 
$\mathcal{T}_L = \{\monomer{m} \in \mathcal{T}: \monomer{m} \cap S_\mathcal{T}^* \neq \emptyset\}$.
Limiting monomers are required to have finite count, whereas nonlimiting monomers (those with all unstarred sites) are allowed to be finite or infinite count.
Our IP representation of a configuration explicitly accounts for \emph{all} the limiting monomers, but only the nonlimiting monomers (in $\mathcal{T} \setminus \mathcal{T}_L$) in a polymer with some limiting monomer;
implicitly every other nonlimiting monomer is unbound (i.e., in its own singleton polymer).
This allows us to describe infinite configurations where all but finitely many of the infinite count monomers are unbound,
guaranteeing that the number of merges counted in \cref{eq:count-merges-to-stable} is finite.

\subsection{Casting StableConfigs as an IP}\label{subsec:formulation}

\subsubsection{Finding a single stable configuration}

We first describe how to find a single element from \functionproblem{StableConfigs}($\mathcal{T}$) by identifying its partial configuration in $\mathcal{T}$.  
We begin by fixing an upper bound $B$ on the number of non-singleton polymers in any partial configuration.  
If no \textit{a priori} bound for $B$ is available, conservatively take $B = |\mathcal{T}_L|$,
the total number of limiting monomers.

\paragraph{Nonnegative integer variables}
Assume an arbitrary ordering of the $m$ monomer types $\monomer{m}_1, \monomer{m}_2, \dots$.  
Our IP formulation uses $B \cdot m + B$  nonnegative integer variables describing the solution via its partial configuration:
\begin{itemize}
    \item $\Count(\monomer{m}, j)$: 
    count of monomer type $\monomer{m} \in \mathcal{T}$ in polymer $\polymer{P}_j$ where $j \in \{1, 2, \dots,B\}$
    \item $\Exists(j)$: false (0) if polymer $\polymer{P}_j$ is empty, possibly true (1) otherwise, $j \in \{1, 2, \dots,B\}$
\end{itemize}

\begin{example}
\label{ex:tbn-ip-formulation-1}
    Recall the TBN of~\cref{fig:tbn-intro-example}.
    Suppose the TBN has 1 each of 
    $\monomer{m}_1 = \{a^*, b^*\},
    \monomer{m}_2 = \{a, b\},
    \monomer{m}_3 = \{a\},
    \monomer{m}_4 = \{b\},$
    with upper bound $B=2$ on the number of non-singleton polymers.
    $\mathcal{T}_L = \{\monomer{m}_1\}$ since $\monomer{m}_1$ is the only monomer with starred sites.
    The linear constraints (see below for details) do not require all copies of $\monomer{m}_2,\monomer{m}_3,\monomer{m}_4 \in \mathcal{T} \setminus \mathcal{T}_L$ to be included in a polymer.
    The stable configuration on the right of~\cref{fig:tbn-intro-example} 
    (partition 
    $\{ \monomer{m}_1, \monomer{m}_2 \}, 
    \{\monomer{m}_3\},
    \{\monomer{m}_4\}$)
    is represented in the IP by setting
    $\Count(\monomer{m}_1, 1) = $ $\Count(\monomer{m}_2, 1) = 1,$
    $\Count(\monomer{m}_3, 1) = $ $\Count(\monomer{m}_4, 1) = 0$
    (monomers 1 and 2 are in polymer 1, but monomers 3 and 4 are not),
    and setting 
    $\Count(\monomer{m}_i, 2) = 0$ for $i= 1,2,3,4$
    (no monomers are in polymer 2),
    $\Exists(1) = 1$
    (polymer 1 is non-empty),
    and 
    $\Exists(2) = 0$
    (polymer 2 is empty).
\end{example}

\begin{example}
    Suppose a TBN with the same monomer types 
    as~\cref{ex:tbn-ip-formulation-1} 
    has 3 of 
    $\monomer{m}_1$
    and infinitely many of the remaining monomers (allowed since they are not limiting),
    with 
    $B=4$. 
    The partial configuration where two copies of $\monomer{m}_1$ are each bound to a single $\monomer{m}_2$ 
    (forming two polymers with two monomers each, as in the stable configuration of~\cref{fig:tbn-intro-example}),
    and the third $\monomer{m}_1$ is bound to an $\monomer{m}_3$ and an $\monomer{m}_4$
    (forming one polymer with three monomers,
    as in the rightmost non-stable saturated configuration of~\cref{fig:tbn-intro-example})
    is represented in the IP by setting
    $\Exists(1) = \Exists(2) = \Exists(3) = 1$ and
    $\Exists(4) = 0$,
    and 
    \[
    \begin{array}{l}
    \Count(\monomer{m}_1, 1) = 1,\quad \Count(\monomer{m}_2, 1) = 1,\quad
    \Count(\monomer{m}_3, 1) = 0,\quad \Count(\monomer{m}_4, 1) = 0,
    \\
    \Count(\monomer{m}_1, 2) = 1,\quad \Count(\monomer{m}_2, 2) = 1,\quad
    \Count(\monomer{m}_3, 2) = 0,\quad \Count(\monomer{m}_4, 2) = 0,
    \\
    \Count(\monomer{m}_1, 3) = 1,\quad \Count(\monomer{m}_2, 3) = 0,\quad
    \Count(\monomer{m}_3, 3) = 1,\quad \Count(\monomer{m}_4, 3) = 1,
    \\
    \Count(\monomer{m}_1, 4) = 0,\quad \Count(\monomer{m}_2, 4) = 0,\quad
    \Count(\monomer{m}_3, 4) = 0,\quad \Count(\monomer{m}_4, 4) = 0.
    \end{array}
    \]
\end{example}

\begin{remark}
The constraints described below allow
$\Exists(j)=0$ even if polymer $j$ is nonempty,
even though the variables ultimately aim to count exactly the number of nonempty polymers 
(as $\sum_{j=1}^m \Exists(j)$).
A false negative undercounts the number of polymers, overcounting the number of merges in~\cref{eq:count-merges-to-stable}. 
However, the number of merges is being \emph{minimized} by the IP.
For a given setting of $\Count$ variables,
the minimum is achieved (subject to the constraints) by setting each $\Exists(j)=1$ \emph{if and only} if polymer $j$ is nonempty.
\end{remark}

\paragraph{Linear constraints}
\todo{DD: I mentioned the constraint above that the variables are nonnegative integers. For completeness it seems we should say that $\Exists(j) \leq 1$ also.}
Let $\mathcal{T}(\monomer{m})$ denote the number of monomers of type $\monomer{m}$ in the TBN $\mathcal{T}$.  
Recall that $\monomer{m}(s)$ is the net count of site type $s \in S_\mathcal{T}$ in monomer type $\monomer{m}$ (negative if $\monomer{m}$ has more $s^*$ than $s$).

\begin{align}
    \Exists(j) 
    &\leq 1
    & \forall j \in \{1,2,\dots,B\}
    \label{exists-leq-1}
    \\
    \sum_{j=1}^{B}\Count(\monomer{m},j) &=    \mathcal{T}(\monomer{m})
     &\forall \monomer{m} \in \mathcal{T}_{L}
     \label{stabletbn-inf-strict-inclusion}
     \\
    \sum_{j=1}^{B}\Count(\monomer{m},j) &\leq \mathcal{T}(\monomer{m})
     &\forall \monomer{m} \in \mathcal{T}\setminus\mathcal{T}_{L}
     \label{stabletbn-inf-bounded-inclusion}
     \\
    \sum_{\monomer{m} \in \mathcal{T}}\Count(\monomer{m}, j)\cdot \monomer{m}(s) &\geq 0
     &\forall j \in \{1,2,\dots,B\}, \forall s \in S_\mathcal{T}
     \label{stabletbn-inf-limiting-sites}
     \\
    \sum_{\monomer{m} \in \mathcal{T}_{L}}\Count(\monomer{m}, j) &\geq \Exists(j)
     &\forall j \in \{1,2,\dots,B\}
     \label{stable-inf-tbn-nonempty}
\end{align}

Constraint~\eqref{exists-leq-1} enforces that $\Exists$ variables are Boolean.
Constraints~\eqref{stabletbn-inf-strict-inclusion} and~\eqref{stabletbn-inf-bounded-inclusion} intuitively establish ``monomer conservation'' in the partial configuration.  
Constraint
\eqref{stabletbn-inf-strict-inclusion} enforces that we account for every limiting monomer in $\mathcal{T}$.
Constraint
\eqref{stabletbn-inf-bounded-inclusion} establishes that for \emph{non}-limiting monomers, we cannot exceed their supply
(trivially satisfied for any infinite-count monomer);
any leftovers are assumed to be in singleton polymers in the full configuration,
but are not explicitly described by $\Count$ variables.
Constraint \eqref{stabletbn-inf-limiting-sites} 
ensures that all polymers are \selfsaturated.  
Specifically, the count of site $s \in S_\mathcal{T}$ within polymer $j$ must meet or exceed that of $s^*$.
Lastly, Constraint \eqref{stable-inf-tbn-nonempty} enforces that 
nonempty polymers contain at least one limiting monomer.
Ideally, this constraint should enforce that if a polymer contains no monomers at all, then it cannot be part of the nonempty polymer tally; however, if the constraint were modeled in this way, the formulation would admit invalid partial configurations that include explicit singleton polymers.

\paragraph{Linear objective function}
Subject to the above constraints, we minimize the number of merges needed to go from a configuration where all monomers are separate to a saturated configuration.
For finite count TBNs, this is the number of monomers minus the number of polymers in the partial configuration.
Equivalently (and applying to infinite TBNs), this is the sum over all nonempty polymers of its number of monomers minus 1.
Formally, the IP minimizes \eqref{eq:objective_value}:
\begin{equation}
\label{eq:objective_value}
    \sum_{j=1}^B
    \qty[ 
        \qty(
            \sum_{\monomer{m} \in \mathcal{T}} \Count(\monomer{m}, j)
        ) - \Exists(j)
    ]
\end{equation}
If 
polymer $j$ is empty
($\sum_{\monomer{m} \in \mathcal{T}} \Count(\monomer{m},j) = 0$),
then constraint~\eqref{stabletbn-inf-limiting-sites} forces $\Exists(j)=0$;
otherwise $\Exists(j)=1$ minimizes~\eqref{eq:objective_value}.
Thus the outer sum is over the nonempty polymers.

\subsubsection{Finding all stable configurations}
While an IP formulation for finding a single stable configuration is well-defined above, without modification it is ill-suited as a formulation to find \textit{all} stable configurations.  
In addition, tightening the available constraints (e.g., enforcing $\Exists(j) \iff$ polymer $j$ is nonempty, described below) provides a more robust framework to which to add custom constraints (e.g. specifying a fixed number of polymers).

\paragraph{IP to find optimal objective value, CP to enumerate optimal solutions}
\label{subsec:ip-find-optimal-then-cp-enumerate}

One straightforward improvement is to solve for the optimal value of the objective function using a dedicated IP solver such as SCIP, whose primal-dual methods exploit the underlying real-valued geometry of the search space to find an objective value more efficiently than Constraint Programming (CP) solvers such as CP-SAT.
Then, use this optimal value to bootstrap the CP formulation, which is better suited to enumerating all solutions with a given objective value.
This works particularly well in our experiments: use SCIP to solve the optimization problem (but SCIP has no built-in ability to enumerate all feasible solutions), then use CP-SAT (which takes longer than SCIP to find the objective value) to locate all feasible solutions to the IP obtaining the objective value found by SCIP.

\paragraph{Enforcing that $\Exists$ variables exactly describe nonempty polymers}
Constraint~\eqref{stabletbn-inf-limiting-sites} enforces that $\Exists(j)=0$ if polymer $\polymer{P}_j$ is empty, 
but it does not enforce the converse.  
However, when using CP-SAT with a \emph{fixed} objective value, we can no longer rely on the minimization of \cref{eq:objective_value} to enforce that $\Exists(j)=1 \iff \polymer{P}_j$ is nonempty.

We add a new constraint to handle this.
Let
\begin{equation}
\label{big_number_definition}
C = 1 + \sum_{s \in S_\mathcal{T}}\sum_{\monomer{m} \in \mathcal{T}} \mathcal{T}(\monomer{m}) \cdot \monomer{m}(s^*).
\end{equation}
$C$ is an upper bound on the largest number of monomers in a polymer in any valid partial configuration of $\mathcal{T}$ minimizing~\cref{eq:objective_value}.  
This corresponds to the worst case in which a single polymer contains \emph{every} limiting monomer, and each starred site is bound to its own unique monomer.  
The following constraint enforces that if $\Exists(j)=0$, then polymer $\polymer{P}_j$ contains no monomers:
\begin{align}
\label{stable-inf-tbn-nonempty-inverse}
    \sum_{\monomer{m} \in \mathcal{T}_{L}}\Count(\monomer{m}, j) &\leq C \cdot \Exists(j)
     &\forall j \in \{1,2,\dots,B\}
\end{align}

\paragraph{Eliminating symmetries due to polymer ordering}
In the formulation of~\cref{subsec:formulation}, many isomorphic solutions exist in the feasible region.
For instance, one could obtain a ``new'' solution by swapping the compositions of polymers $\polymer{P}_1$ and $\polymer{P}_2$.  
The number of isomorphic partial configurations grows factorially with the number of polymers.  
Before asking the solver to enumerate all solutions,
we must add constraints that eliminate isomorphic solutions.  
We achieve this by using the (arbitrary) ordering of the monomer types to induce a lexicographical ordering on the polymers,  
then add constraints ensuring that any valid solution contains the polymers in sorted order.

Sorting non-binary vectors in an IP is generally a difficult task (for instance, see \cite{ip-sorting}).  The primary reason for this difficulty is that encoding the sorting constraints involves logical implications ($p \implies q$), which, being a type of disjunction ($\neg p$ OR $q$), are difficult to encode into a convex formulation described as a conjunction (AND) of several constraints.
However, we do have an upper bound $C$ on the values that the $\Count$ variables can take, making certain ``large-number'' techniques possible.

Intuitively, when comparing two lists of scalars (i.e., vectors) to verify that they are correctly sorted, one must proceed down the list of entries until one of the entries is larger than its corresponding entry in the other list.  For as long as the numbers are the same, they are considered ``tied''.  When one entry exceeds the corresponding other, the tie is considered ``broken'', after which no further comparisons need be conducted between the two vectors.  

We $B \cdot m$ new Boolean (0/1-valued) variables ($\Tied(\monomer{m}_i, j)$ for each $i=1,\dots,m$ and $j=1,\dots,B$),
that reason about consecutive pairs of polymers $\polymer{P}_{j-1}$, $\polymer{P}_{j}$.
We describe constraints enforcing that for each $h \leq i$,
    $\Tied(\monomer{m}_i, j) = 1 \iff \Count(\monomer{m}_h, j-1) 
     = \Count(\monomer{m}_h, j).$

Let $C$ be defined as in \eqref{big_number_definition}. 
For simplicity of notation below,
define the constants $\Tied(\monomer{m}_0, j) = 1$ for all $j=1,\dots,B$.
The meaning of the sorting variables is then enforced by the following constraints, which we define for $i \in \{1, 2,\dots,m\}$ and $j \in \{2,3,\dots,B\}$:
\begin{align}
    \Tied(\monomer{m}_i, j) &\leq 
     \Tied(\monomer{m}_{i-1}, j)\label{stabletbn-tiebreak-grandfather}\\
    \Count(\monomer{m}_i, j-1) -
     \Count(\monomer{m}_i, j) &\leq C \cdot (1-
     \Tied(\monomer{m}_i, j))\label{stabletbn-tiebreak-equal-part1}\\
    \Count(\monomer{m}_i, j-1) -
     \Count(\monomer{m}_i, j) &\geq -C \cdot (1-
     \Tied(\monomer{m}_i, j))\label{stabletbn-tiebreak-equal-part2}\\
    \Count(\monomer{m}_i, j-1) -
     \Count(\monomer{m}_i, j)
     &\geq 1 - C \cdot \qty\big(
     1+ \Tied(\monomer{m}_i, j)-\Tied(\monomer{m}_{i-1}, j)
     )\label{stabletbn-tiebreak-broken}
\end{align}

Intuitively, \eqref{stabletbn-tiebreak-grandfather} enforces
$\Tied(\monomer{m}_i, j) \implies \Tied(\monomer{m}_{i-1}, j)$:
a tie in the current entry is only relevant if the tie was not resolved before. \eqref{stabletbn-tiebreak-equal-part1} and \eqref{stabletbn-tiebreak-equal-part2} together enforce
$\Tied(\monomer{m}_i, j) \implies \qty\big(\Count(\monomer{m}_i, j-1) =
\Count(\monomer{m}_i, j))$:
ties can only continue for as long as the corresponding entries are equal.
\eqref{stabletbn-tiebreak-broken} enforces
$\Tied(\monomer{m}_{i-1}, j)
\land
\lnot\Tied(\monomer{m}_i, j)
\implies 
\qty\big(\Count(\monomer{m}_i, j-1) >
\Count(\monomer{m}_i, j))$:
ties can only be broken if the tie was not broken previously and the current entries are ordered correctly.
Thus any solution satisfying these constraints must sort the polymers.

\subsection{Empirical running time measurements}\label{subsec:benchmarks}

\begin{figure}
    \centering
    \vspace{-0.5cm}
    \begin{subfigure}{.56\textwidth}
        \centering
        \includegraphics[width=\linewidth,
        trim=1.8cm 0.2cm 0cm 0cm,clip
        ]{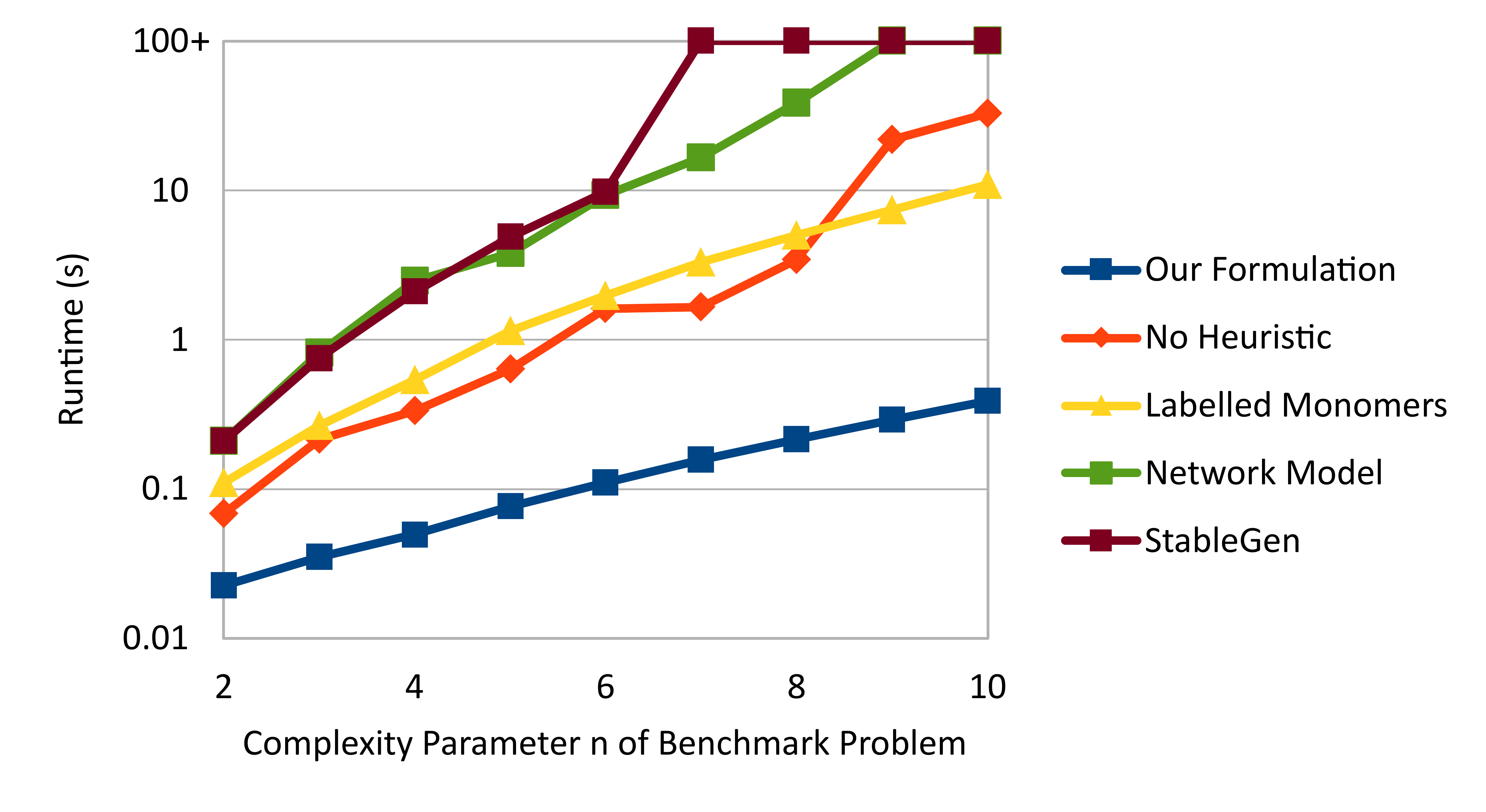}
        \label{fig:benchmark_by_grid_size}
    \end{subfigure}%
    \begin{subfigure}{.46\textwidth}
        \centering
        \includegraphics[width=\linewidth,
        trim=1.0cm 0.2cm 0cm 0cm,clip
        ]{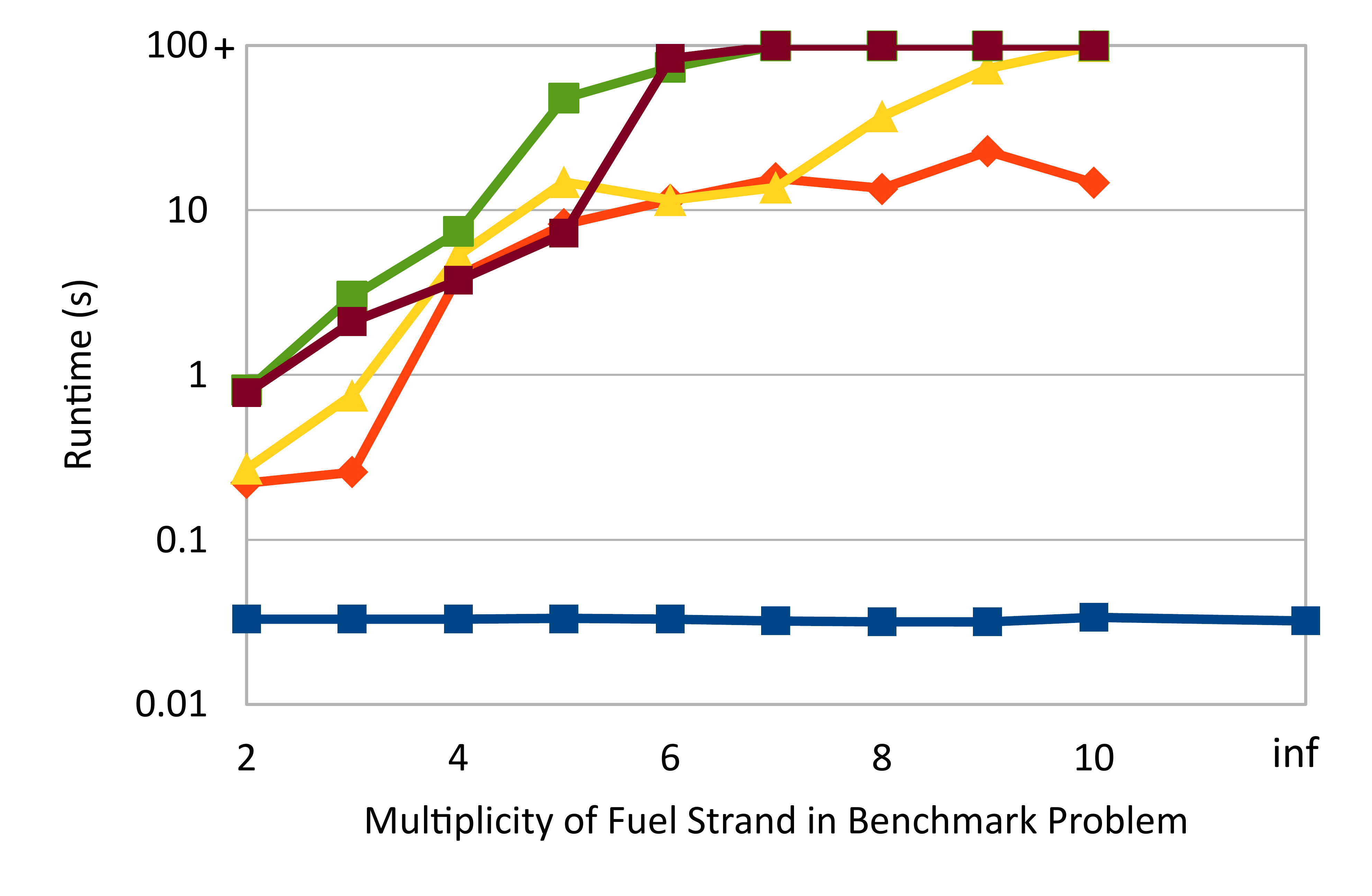}
        \label{fig:benchmark_by_fuel_quantity}
    \end{subfigure}
    \vspace{-0.6cm}
    \caption{\footnotesize 
    Empirical tests solving \functionproblem{StableConfigs} for our benchmark problem based upon its complexity parameter $n$ (left), and the multiplicity of the unstarred ``fuel'' strands (right).  Our formulation is tested against several variations on the approach (which are described in the text) and the StableGen algorithm from \cite{tbn-sat}.
    The TBN is parameterized by $n$ and contains the monomers
    $G_n = \{ x_{ij}^* : 1 \leq i,j \leq n \}$,
    $H_i = \{ x_{ij} : 1 \leq j \leq n \}$ for all $1 \leq i \leq n$, and
    $V_j = \{ x_{ij} : 1 \leq i \leq n \} \cup \{x_{ij} : j \leq i \leq n\}$ for all $1 \leq j \leq n$. 
    See Fig. 6 from~\cite{ktbn} for a detailed explanation of this TBN and its operation.
    The vertical axis is log scale.  Points at the top of the scale timed out after 100 seconds.  The alternate formulations cannot solve the instance in the case of infinite fuel strands.}
    \label{fig:benchmarks}
    \vspace{-0.5cm}
\end{figure}



For our empirical tests we use as a benchmark the autocatalytic TBN described in \cite[Section 4.2.2 and Fig. 6]{ktbn}.  This TBN features two large monomers of size $n^2$ in which $n$ is a parameter in the design, as well as a variable number of additional monomers (``fuels'') intended to be present in large excess quantities.

In addition to the formulation we give in this paper, we also tested a number of formulation variants, including the StableGen algorithm originally posed in \cite{tbn-sat} for solving the \functionproblem{StableConfigs} problem, justifying some of our design choices.   ``No Heuristic'' performs a thorough accounting of all monomers (not just those needed to achieve saturation against the limiting monomers).  ``Labelled Monomers'' assumes that the monomers are provided as a set, rather than a multiset.  ``Network Model'' is a modification of StableGen with an alternate saturation constraint which does not require the explicit invocation of site-level bonds.

Each data point represents the average of three runs, and the solver was allowed to run for up to 100 seconds before a timeout was forced.
\cref{fig:benchmarks} (left) shows the runtimes as they increase with the parameter $n$, holding the count of each fuel at 2.
\cref{fig:benchmarks} (right) fixes $n=3$ and shows the runtimes as they increase with the multiplicity of the fuel monomers.  Note that our formulation can solve the case when fuels are in unbounded excess, while the variant formulations require bounded counts of all monomers.

Our formulation solves all of the benchmark problems in under one second, suggesting that it is suitable for much larger/more complex problems than were approachable previously.


\section{Computing bases of locally stable configurations of TBNs}
\label{sec:hilbert}




We now shift attention to \emph{locally} stable configurations:
those in which no polymer can be split without breaking a bond.
Such a configuration may not be stable, but the only paths to create more polymers,
without breaking any bonds,
require first merging existing polymers (i.e., going uphill in energy before going down).
The saturated configurations are precisely those obtained by merging polymers starting from some locally stable configuration.
In this section we describe a technique for computing what we call the \emph{polymer basis}:
the (finite) set of polymers that can exist in locally stable configurations.
In~\cref{subsec:polymer-basis-hilbert-basis-equiv},
we show that an algebraic concept called the \emph{Hilbert basis}~\cite{de2012algebraic} characterizes the polymer basis.
In~\cref{subsec:polymer-basis-reason-behavior,subsec:translator-cycle} we show how the polymer basis can be used to reason about TBN behavior.


\subsection{Equivalence of polymer bases and Hilbert bases}
\label{subsec:polymer-basis-hilbert-basis-equiv}

We note that the connection between Hilbert bases and polymer bases is not particularly deep and does not require clever techniques to prove.
Once the definitions are appropriately set up, 
the equivalence follows almost immediately.
(Though we provide a self-contained proof.)
The primary insight of this section is that
casting TBNs in our IP formulation sets up the connection with Hilbert bases.
Since highly optimized software exists for computing Hilbert bases~\cite{4ti2},
this software can be deployed to automate reasoning about TBNs.

Let $M$ be a set of monomer types with $m = |M|$.
Let $\mathcal{S}_M$ denote the \emph{TBN schema of $M$},
the set of all TBNs containing only monomers from $M$,
such that starred sites are limiting
(i.e., such that saturated configurations have all starred sites bound).
Let $A_M$ be the \emph{matrix representation} of the monomer types in $\mathcal{S}_M$, 
describing the contents of each monomer type: formally, the row-$i$, column-$j$ entry of $A_M$ is $\monomer{m}_j(s_i)$, the net count of site type $s_i$ in monomer type $\monomer{m}_j$ (as an example, 
$\{a^*, b, a, a, a, c, c^*, c^*\}$ has net count $2$ of $a$, $1$ of $b$, and $-1$ of $c$).
Formally, a TBN $\mathcal{T} \in \mathcal{S}_M$ if and only if 
$A_M \mathcal{T} \geq \vec{0}$.

Recall that $\saturatedconfigs{\mathcal{T}}$ is the set of saturated configurations of the TBN $\mathcal{T}$,
and that a polymer $\polymer{p}$ is \emph{\selfsaturated}\ if it has no exposed starred sites,
i.e., $A_M \polymer{p} \geq \vec{0}$.
Define the \emph{polymer basis} $\mathcal{B}_{\mathcal{S}_M}$ to be the set of all polymers $\polymer{P}$ with the following properties:
\begin{itemize}
    \item
    $(\exists \mathcal{T} \in \mathcal{S}_M)
    (\exists \alpha \in \saturatedconfigs{\mathcal{T}})\ 
    \polymer{P} \in \alpha$
    (i.e., $\polymer{P}$ appears in some saturated configuration of a TBN using only the monomer types from $M$.)
    
    \item 
    There is no partition of $\polymer{P}$ into two (or more) \selfsaturated\ polymers.
\end{itemize}
For example, consider the monomers 
$G = \{a^*,b^*,c^*,d^*\},$
$H_1 = \{a,b\},$
$H_2 = \{c,d\},$
$V_1 = \{a,c\},$
$V_2 = \{b,d\}$
and let $M = \{G,H_1,H_2,V_1,V_2\}$.
The polymer basis 
$\mathcal{B}_{\mathcal{S}_M}$ is
$\{$
$\{G, H_1, H_2\},$
$\{G, V_1, V_2\},$
$\{H_1\},$
$\{H_2\},$
$\{V_1\},$
$\{V_2\}$
$\}$.
All other \selfsaturated\ polymers are unions of these.


To show that polymer bases can be characterized by Hilbert bases, we must first define some additional terms.
A \emph{conical combination} of a set of vectors is a linear combination of the vectors using only nonnegative coefficients.
An \emph{integer conical combination} of a set of vectors is a conical combination of the vectors using only integer coefficients.
A \emph{(polyhedral) convex cone} $C = \{\lambda_1 \vec{a}_1 + \dots + \lambda_n \vec{a}_n : \lambda_1, \dots, \lambda_n \geq 0\}$ is the space of all conical combinations of a finite set of vectors $\{\vec{a}_1, \dots, \vec{a}_n\}$ (and is said to be \emph{generated} by $\{\vec{a}_1, \dots, \vec{a}_n\}$).
$C$ is \emph{pointed} if $C \cap (-C) = \{\vec{0}\}$.
A set of the form $\{\vec{x} \in \mathbb{R}^m: A\vec{x} \geq \vec{0} \text{ and } \vec{x} \geq \vec{0}\}$ is always a pointed convex cone~\cite{de2012algebraic}.

A set is \emph{inclusion-minimal} with respect to a property if it has no proper subset that satisfies the same property.
The \emph{Hilbert basis} of a pointed convex cone $C$ is the unique inclusion-minimal set of integer vectors such that every integer vector in $C$ is an integer conical combination of the vectors in the Hilbert basis.
For example, the Hilbert basis of the convex cone generated (with nonnegative real coefficients) by $(1,3)$ and $(2,1)$ is $\{(1,1), (1,2), (1,3), (2,1)\}$;
note that $\frac{2}{5} \cdot (1,3) + \frac{4}{5} \cdot (2,1) = (2,2)$, which is not an integer combination of $(1,3)$ and $(2,1)$,
but $2 \cdot (1,1) = (2,2)$.

Recall that the matrix-vector product $A_M\polymer{p}$ gives the number of exposed sites of each type in the polymer, so that $A_M\polymer{p} \geq \vec{0}$ iff the polymer is \selfsaturated\ (i.e. none of the starred sites are exposed).




We are then interested in vectors contained in 
$\{\polymer{p} \in \N^m: A_M\polymer{p} \geq \vec{0}\}$.
Noting that $\N^m = \{\polymer{p} \in \mathbb{R}^m: \polymer{p} \geq \vec{0}\}\cap \mathbb{Z}^m$,
we can equivalently state that we are interested in all integer vectors contained in the pointed convex cone 
$\{\polymer{p} \in \mathbb{R}^m: A_M\polymer{p} \geq \vec{0}
\text{ and }
\polymer{p} \geq \vec{0}\}$.

\begin{theorem} \label{thm:polymer_basis_is_hilbert_basis}
    Let $\mathcal{S}_M$ be a TBN schema and let $A_M$ be the matrix representation of its monomer types.
    Then the polymer basis $\mathcal{B}_{\mathcal{S}_M}$ of $\mathcal{S}_M$ is the Hilbert basis of $\{\polymer{p} \in \mathbb{R}^m: A_M\polymer{p} \geq \vec{0} \text{ and } \polymer{p} \geq \vec{0}\}$.
\end{theorem}
\begin{proof}
    Note that the integer vectors in $\{\polymer{p} \in \mathbb{R}^m: A_M\polymer{p} \geq \vec{0} \text{ and } \polymer{p} \geq \vec{0}\}$ are precisely the polymers that appear in saturated configurations of TBNs in $\mathcal{S}_M$,
    since $A_M\polymer{p} \geq \vec{0} \iff$ polymer $\polymer{p}$ is \selfsaturated,
    and $\mathcal{S}_M$ is defined to have starred sites limiting, so that a configuration is saturated if and only if each of its polymers is \selfsaturated.
    
    We must show two properties to establish that $\mathcal{B}_{\mathcal{S}_M}$ is the Hilbert basis.
    First we must show that every polymer in saturated configurations of $\mathcal{S}_M$ is a nonnegative integer combination of polymers in $\mathcal{B}_{\mathcal{S}_M}$.
    Next, to establish inclusion-minimality, we must show that no polymer can be removed from $\mathcal{B}_{\mathcal{S}_M}$ while satisfying the first property.
    
    To see the first property,
    consider a polymer $\polymer{P}$ in a saturated configuration of some TBN in $\mathcal{S}_M$.
    If it cannot be split into multiple self-saturated polymers, then we are done since it is in $\mathcal{B}_{\mathcal{S}_M}$
    (it is the integer combination consisting of one copy of itself).
    Otherwise, we can iteratively split $\polymer{P}$ into polymers $\polymer{P}_1,\dots,\polymer{P}_k$ that themselves cannot be split into self-saturated polymers.
    Then $\polymer{P} = \polymer{P}_1 + \dots + \polymer{P}_k$.
    
    To see the second property,
    consider a \selfsaturated\ polymer $\polymer{P} \in \mathcal{B}_{\mathcal{S}_M}$ that can be removed while maintaining the first property.
    Since $\polymer{P}$ is an integer vector in 
    $\{\polymer{p} \in \mathbb{R}^m: A_M\polymer{p} \geq \vec{0} \text{ and } \polymer{p} \geq \vec{0}\}$, 
    $\polymer{P}$ is the nonnegative integer sum
    of some polymers remaining in $\mathcal{B}_{\mathcal{S}_M} \setminus \{\polymer{P}\}$.
    However, all polymers in $\mathcal{B}_{\mathcal{S}_M}$ are self-saturated,
    so $\polymer{P}$ can be partitioned into multiple self-saturated polymers.
    Thus $\polymer{P}$ is not an element of $\mathcal{B}_{\mathcal{S}_M}$ to begin with.
\end{proof}

\subsection{Using the polymer basis to reason about TBN behavior}
\label{subsec:polymer-basis-reason-behavior}

The complexity of computing the polymer basis in general can be very large; however, once it is calculated, reasoning about the stable configurations becomes a simpler task.
For instance, in a previous example we had
$\mathcal{B}_{\mathcal{S}_M} = $
$\{$
$\{G, H_1, H_2\},$
$\{G, V_1, V_2\},$
$\{H_1\},$
$\{H_2\},$
$\{V_1\},$
$\{V_2\}$
$\}$.
We can see from the above basis that in saturated configurations, $G$ can only be present one of two unsplittable polymer types: $\{G, H_1, H_2\}$ or $\{G, V_1, V_2\}$, and we can optimize the number of polymers in a configuration by taking the other monomers as singletons (which is allowed, as these singletons are in the polymer basis).
More generally, reasoning about stable configurations amounts to determining the number of each polymer type to use from the polymer basis so that the union of all polymers is the TBN, while using the maximum number of polymers possible.  Our software can also solve for stable configurations in this way; specifically, for a TBN $\mathcal{T}$, it can calculate the polymer basis (abbreviated here as $\mathcal{B}$) and then solve for the stable configurations using the following IP:

$$
\max_{
    \vec{c} \in
    \mathbb{N}^{
        \abs{
            \mathcal{B}
        }
    }
}
\norm{\vec{c}}_1
\text{  subject to  }
\sum_{i=1}^{
    \abs{\mathcal{B}}
}
c_i \mathcal{B}_i = \mathcal{T}
$$
Alternately, one can solve for the stable systems via an augmentation approach (see~\cite{de2012algebraic}).


If the goal is simply to solve the \functionproblem{StableConfigs} problem, we do not expect that solving for the stable configurations in this way will be more efficient than the previous formulation, as the time spent computing the Hilbert basis alone can require a great deal longer than solving via the formulation of the previous section.  Instead,
the true value of the basis is in its ability to describe \textit{all} saturated configurations of a TBN.

For instance, in \cite{ktbn}, the authors define an augmented TBN model in which a system can move between saturated configurations by two atomic operations: polymers can be pairwise merged (with an energetic penalty, i.e., higher energy) or they can be split into two so long as no bonds are broken (with an energetic benefit, i.e., lower energy; for instance $\{a,b\}, \{a^*, b^*\},  \{a\}, \{a^*\}$ can be split into $\{a, b\}, \{a^*, b^*\}$ and $\{a\}, \{a^*\}$, whereas $\{a\}, \{a^*\}$ cannot be split).
Any saturated polymer not in the basis can split into its basis components without breaking any bonds. 
Thus the polymer basis contains all polymers that can form in a \emph{local} minimum energy configuration, i.e., one where no polymer can split.

When designing a TBN, the designer will typically have a sense for which polymers are to be ``allowed'' in local energy minima. 
Proving that the system observes this behavior was not previously straightforward, but we can now observe that the TBN will behave ideally when its expected behavior matches its polymer basis.


\subsection{A case example: Circular Translator Cascade}
\label{subsec:translator-cycle}

We now discuss an example of using the polymer basis to reason about a TBN's kinetic behavior, studying a TBN known as a \emph{circular translator cascade}, first defined in~\cite{ktbn}:
\begin{align*}
\{&   
    \{a,b,c\},
    \{b,c,d\},
    \{c,d,e\},
    \{d,e,f\},
    \{e,f,a\},
    \{f,a,b\},
    \\&
    \{a*,b*\},
    \{b*,c*\},
    \{c*,d*\},
    \{d*,e*\},
    \{e*,f*\},
    \{f*,a*\}
\}
\end{align*}
There are two stable configurations of this TBN, shown in~\cref{fig:translator_cascade}.
\begin{figure}[h]
    \vspace{-0.2cm}
    \centering
    \includegraphics[height=4cm]{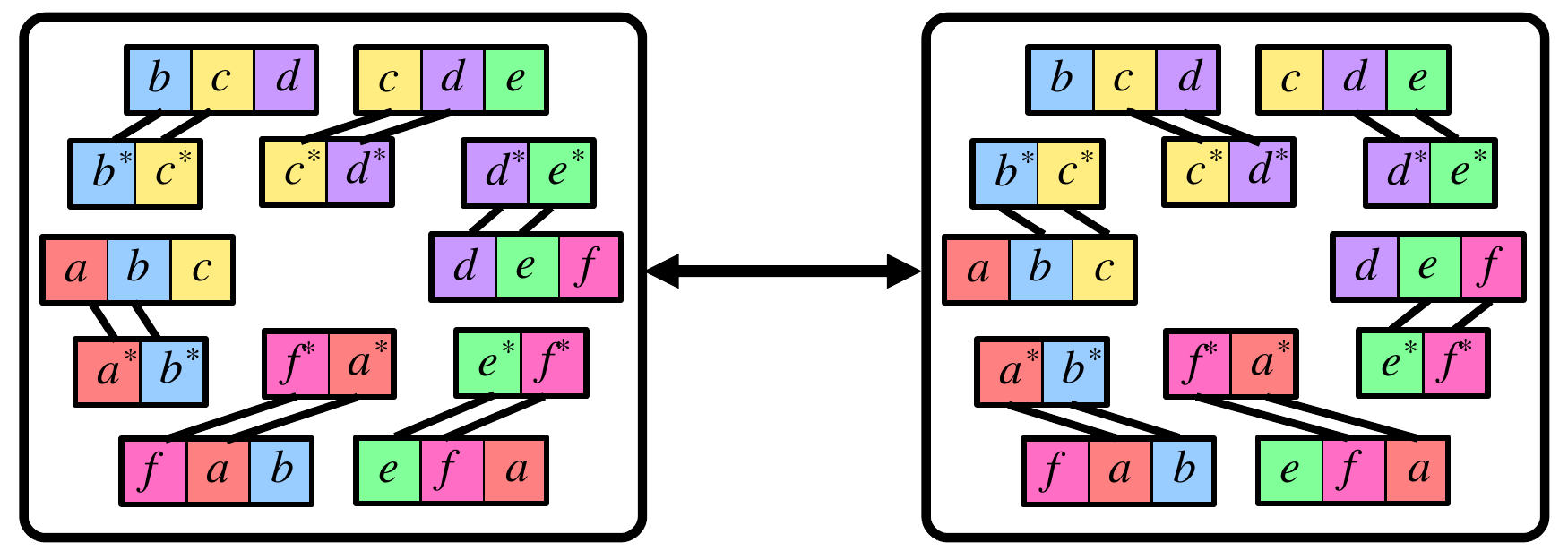}
    \vspace{-0.2cm}
    \caption{ \footnotesize
        The two stable configurations of a variant of the circular translator cascade described in \cite{ktbn}.  In the left configuration, the unstarred monomers are bound to their ``left-side'' companions (e.g. $\{a,b,c\}$ is bound to $\{a^*, b^*\}$), and in the right configuration, the unstarred monomers are bound to their ``right-side'' companions (e.g. $\{a,b,c\}$ is bound to $\{b^*, c^*\}$).
    }
    \label{fig:translator_cascade}
    \vspace{-0.3cm}
\end{figure}

We consider now the ``pathways'' by which one of the stable configurations can ``transition'' to another.  This process is described formally in \cite{ktbn}; here we give an intuitive description.  Informally, we admit as atomic operations the ability for two polymers to merge or for one polymer to split into two polymers, so long as the resulting configuration remains saturated. In essence, these operations are modelling the physical phenomenon of solutes colocalizing in solution before reactions occur, specifically in dilute solutions in which enthalpic bond rearrangements occur on a timescale much faster than the timescale for entropic colocalization.  
If many polymers must be merged in some intermediate configuration to transition between stable configurations, then since each merge is individually unlikely, the successive merges required are exponentially unlikely
i.e., a large \emph{energy barrier} exists to transition between the configurations.

The design intention of this TBN is to have two stable configurations with a large energy barrier to transition between them.
For the largest possible energy barrier, the transition should require the simultaneous merging of \emph{all} of the polymers into a single polymer as an intermediate step.  
However, this is not the case for the TBN of~\cref{fig:translator_cascade}; the polymer basis gives insight into why.
See~\cite[Section A.2]{ktbn} for an argument why more domain types and monomer types are required.
We interpret this as a design error (in fact it actually \emph{was} a design error in an early draft of~\cite{ktbn}).
We now explain how the error can be detected by reasoning about the polymer basis of the system,
justifying that the automated computation of the polymer basis by our software enables one to automate some reasoning about the correct behavior of TBNs.



If it were true that the polymer basis contained only the 12 polymer types that are present in the two stable configurations of~\cref{fig:translator_cascade}, 
then that would be sufficient to prove the high energy barrier.
To see why this is true, suppose there were a locally stable intermediate configuration that is part of a lower barrier transition.  Since the configuration is locally stable, it is saturated, and none of its polymers can be partitioned into self-saturated polymers.  By definition, the polymer basis should then contain all of the polymers present in this intermediate configuration.  However, all of the polymers in the basis have exactly two monomers, and so there must be 6 polymers in the intermediate configuration.  The stable configurations also have 6 polymers, and so the intermediate configuration is also stable, but this contradicts that there are only two stable configurations.

\begin{figure}[ht]
    \vspace{-0.25cm}
    \centering
    \includegraphics[width=5.6in, 
    trim=0.6cm 0.6cm 0.6cm 0.6cm, 
    clip]{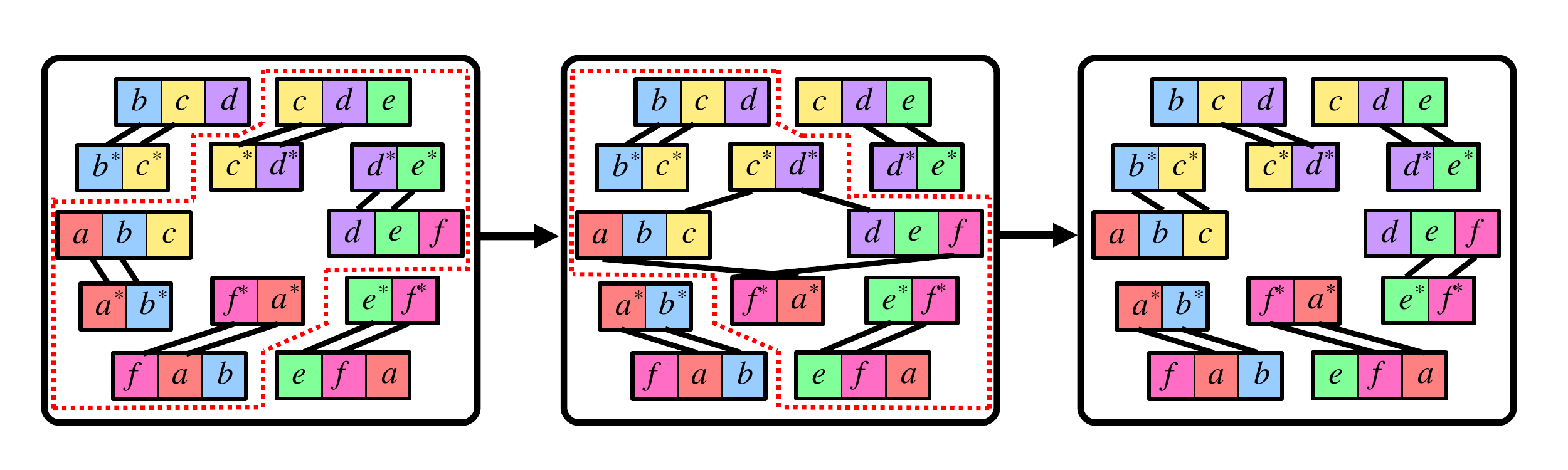}
    \vspace{-0.6cm}
    \caption{ \footnotesize
        A counterexample to the claim that transitioning between stable configurations of this TBN requires the simultaneous merger of all monomers into a single polymer.  Starting from the stable configuration on the left, by merging the four polymers within the red dotted outline, it is possible to re-arrange bonds and then split to the middle configuration.  Then from the middle configuration, by merging the three polymers in the red dotted online, it is possible to re-arrange bonds and then split to the stable configuration on the right. Such intermediate configurations are evident by examining the elements of the polymer basis.
    }
    \label{fig:translator_cascade_defeat}
    \vspace{-0.2cm}
\end{figure}

In fact, the polymer basis for this TBN has 57 entries (determined via our software),
not 12, and we can use this basis to disprove the high energy barrier,
i.e, to show that there is a sequence of merges and splits that transitions between the two stable configurations,
without all monomers ever being merged into a single polymer.
To discover a pathway that demonstrates the lower energy barrier, consider one unexpected entry in the polymer basis: $\polymer{P} = \{\{a,b,c\}, \{d,e,f\},\{c^*, d^*\},\{f^*,a^*\}\}$.  Its existence in the polymer basis tells us that there must be some saturated configuration that contains it.  If we examine where these monomers were in one of the original stable configurations (\cref{fig:translator_cascade}, left), we see that these were originally in polymers
\[
\{\{a,b,c\},\{a^*, b^*\}\},\quad
\{\{c,d,e\},\{c^*, d^*\}\},\quad
\{\{d,e,f\},\{d^*, e^*\}\}, \quad
\{\{f,a,b\},\{f^*, a^*\}\}.
\]
From the starting configuration, if only these four polymers were merged, then they could then iteratively split into $\polymer{P}$, $\{\{c,d,e\},\{d^*,e^*\}\}$, and $\{\{f,a,b\},\{a^*,b^*\}\}$.  Since the latter two polymers are part of the target configuration,
one could now greedily merge all polymers except for these latter two and then split into the target configuration.  At no point in the interim were all polymers merged together into a single polymer.  The resulting pathway is illustrated in~\cref{fig:translator_cascade_defeat}.

The difference between intended and actual barrier in this design becomes more pronounced if it is scaled up to include more site types and monomers.  
In~\cite{ktbn} it is shown that by modifying the design, it is possible to achieve a linear energy barrier by using a quadratic number of site types.  

\section{Conclusion}


In our investigation we observed that it was generally more efficient to solve \functionproblem{SaturatedConfigs} by finding the optimal objective value using an IP solver as a first step, followed by using a CP solver on the same formulation with the objective value now constrained to the value found by the IP solver.
Are further computational speedups possible by using IP as a callback during the CP search, instead of only in the beginning?  How would one formulate the subproblems that would need to be solved in these callbacks?

In this paper we also note the value of polymer bases that are derived from the matrix containing the monomer descriptions.  Such polymer bases can be used to describe all saturated configurations of a TBN, and so provide a valuable tool for analyzing potential behavior of a TBN when the model is augmented with rules that allow for dynamics.  In practice, rather than discover unexpected behavior by calculating the polymer basis, a designer would instead like to begin with a set of behaviors and then create a TBN that respects them.  Can we begin from verifiable polymer/Hilbert bases, encoding desired behavior, and transform them into TBN/DNA designs?

The full TBN model \cite{ktbn} can also be used to describe the regime in which there is a more modest tradeoff between the two energetic penalties of breaking bonds and merging complexes (i.e., saturation is not guaranteed).
For example toehold-length binding sites in DNA strand displacement systems~\cite{qian2011scaling, srinivas2017enzyme, chen2013programmable, zhang2009control, zhang2011dynamic}
are intended to dissociate over timescales comparable to association.
Indeed, our software \cite{stable_tbn_software} includes an implementation
of the \functionproblem{StableConfigs} formulation in which this relative weighting factor is included in the objective function.  Under what conditions can a comparable polymer basis for such a system be found?  Within the context of integer programming, it is known that by adding constraints to the design, one can reduce the complexity of finding/verifying Hilbert bases (and related \emph{Graver bases}) \cite{best_algorithms_for_graver}, but it is not clear how to interpret these numeric constraints within the context of TBNs.

Satisfiability Modulo Theory (SMT) formulations
have the ability to express TBN concepts (such as reachability along kinetic paths) without converting to a linear algebra framework, and existing solvers can solve SMT instances with surprising efficiency (for example, Z3 \cite{z3}).  Can such solvers reason about TBNs directly within a reasonable time frame?  Can they efficiently extract information beyond what is contained in the polymer basis?


\bibliography{references}
\bibliographystyle{plainurl}

\end{document}